\documentclass[12pt]{article}
\usepackage{fullpage,amssymb,graphicx,url,amsmath}

\newtheorem{property}{Property}
\newtheorem{definition}{Definition}
\newtheorem{example}{Example}
\newtheorem{proposition}{Proposition}
\newtheorem{remark}{Remark}

\def\calL{\mathcal{L}}
\def\GL{\mathrm{GL}}

\def\barf{{\bar{f}}}
\def\barg{{\bar{g}}}

\def\barF{{\bar{F}}}
\def\barG{{\bar{G}}}

\def\bartheta{{\bar\theta}}
\def\bareta{{\bar\eta}}

\def\barx{{\bar x}}
\def\bary{{\bar y}}

\def\bbR{\mathbb{R}}

\def\inner#1#2{{\langle #1,#2\rangle}}

\def\vecthree#1#2#3{{\left[\begin{array}{c}#1\cr #2\cr #3\end{array}\right]}}

\newenvironment{proof}{\paragraph{Proof:}}{\hfill$\square$}

\def\SepAnd{;\ }
\def\keywords#1{\vskip 0.5cm\noindent {\bf Keywords}: #1}

\graphicspath{{Fig/}{./}}

\title{Beyond scalar quasi-arithmetic means: Quasi-arithmetic averages and quasi-arithmetic mixtures in information geometry\footnote{A preliminary version appeared in~\cite{qac-2023} with technical report~\protect\url{http://arxiv.org/abs/2301.10980}}}
 
\author{Frank Nielsen\\ \ \\ Sony Computer Science Laboratories Inc.\\ Tokyo, Japan}
\date{}

\begin{document}

\sloppy
\def\calP{\mathcal{P}}
\def\bareta{{\bar\eta}}
\def\calM{\mathcal{M}}
\def\bbX{\mathbb{X}}
\def\bbR{\mathbb{R}}
\def\barG{{\bar G}}
\def\bartheta{{\bar\theta}}
\def\barF{{\bar F}}
\def\calC{\mathcal{C}}
\def\GL{\mathrm{GL}}
\def\calG{\mathcal{G}}
\def\LSE{\mathrm{LSE}}
\def\KL{\mathrm{KL}}
\def\JS{\mathrm{JS}}
\def\barM{{\bar M}}
\def\calE{\mathcal{E}}
\def\calX{\mathcal{X}}
\def\calF{\mathcal{F}}
\def\dP{\mathrm{d}P}
\def\dQ{\mathrm{d}Q}
\def\dmu{\mathrm{d}\mu}
\def\bbR{\mathbb{R}}
\def\du{\mathrm{d}u}
\def\SPD{\mathrm{Sym}_{++}(d)}
\def\det{\mathrm{det}}
\def\tr{\mathrm{tr}}
\def\inner#1#2{\left\langle #1,#2 \right\rangle}
\def\idfunc{\mathrm{id}}
\def\vectortwo#1#2{{\left[\begin{array}{ll}#1 \cr #2\end{array}\right]}}
\def\CM{\mathcal{CM}}
\def\dlambda{\mathrm{d}\lambda}
\def\vectorcol#1#2#3{\left[\begin{array}{c}#1\cr #2\cr #3\end{array}\right]}
\def\st{\ :\ }
\def\etabar{\bar\eta}
\def\DFS{\mathrm{DFS}}
\def\Sym{\mathrm{Sym}}
\def\IE{\mathrm{IE}}
\def\labelg#1{\stackrel{#1}{g}}
\def\qam{m}
\def\qaa{M}
\def\qamix{m}

\maketitle              
\begin{abstract}
We generalize quasi-arithmetic means beyond scalars by considering continuously invertible gradient maps of strictly convex  Legendre type real-valued functions. 
Gradient maps of strictly convex Legendre type functions are strictly comonotone and admits a global inverse,
thus   generalizing the notion of strictly mononotone and differentiable functions used to define scalar quasi-arithmetic means.
Furthermore, the Legendre transformation gives rise to pairs of {\em dual quasi-arithmetic averages} via the convex duality.
We study both the invariance and equivariance properties under affine transformations of quasi-arithmetic averages via the lens of dually flat spaces of information geometry. 
We show how these quasi-arithmetic averages are used to express points on dual geodesics and sided barycenters in the dual affine coordinate systems.
Finally, we consider quasi-arithmetic mixtures and describe several parametric and non-parametric statistical models which are closed under the quasi-arithmetic mixture operation.

\keywords{quasi-arithmetic mean \SepAnd  Legendre transform \SepAnd Legendre-type function \SepAnd information geometry \SepAnd affine Legendre invariance  \SepAnd Jensen divergence \SepAnd comparative convexity \SepAnd Jensen-Shannon divergence}
\end{abstract}
\section{Introduction}


We first start by generalizing the notion of quasi-arithmetic means~\cite{Inequalities-1952} (Definition~\ref{def:qam}) which relies on strictly monotone and differentiable functions to other non-scalar types such as vectors or matrices in Section~\ref{sec:qaaig}: 
Namely, we show how the gradient of a strictly convex and differentiable function of Legendre type~\cite{Rockafellar-1967} (Definition~\ref{def:Legendre}) is co-monotone (Proposition~\ref{prop:gradcomonotone}) and admits a continuous global inverse.
Legendre type functions bring the counterpart notion of quasi-arithmetic mean generators to non-scalar types that we term 
{\em quasi-arithmetic averages} (Definition~\ref{def:mqam}).
In Section~\ref{sec:qaaigd}, we show how quasi-arithmetic averages occur naturally in the dually flat manifolds of information geometry~\cite{amari1985differential,IG-2016}:
Quasi-arithmetic averages are used to express the coordinates of (1) points on dual geodesics (\S\ref{sec:geo}) and (2)
dual barycenters with respect to the canonical divergence which amounts to a Bregman divergence~\cite{amari1985differential} (\S\ref{sec:bary}).
We explain the dualities between steep exponential families~\cite{BN-2014}, regular Bregman divergences~\cite{banerjee2005clustering}, and quasi-arithmetic averages in Section~\ref{sec:dualQAA} and interpret the calculation of the induced geometric matrix mean using quasi-arithmetic averages in Section~\ref{sec:inductive}.
The invariance and equivariance properties of quasi-arithmetic averages are studied in Section~\ref{sec:invariance}
 under the framework of information geometry: 
The invariance and equivariance of quasi-arithmetic averages under affine transformations (Proposition~\ref{prop:equiinvariance}) generalizes the invariance property of quasi-arithmetic means (Property~\ref{propr:invarqam}) and bring new insights from the
  information-geometric viewpoint.
Finally, in Section~\ref{sec:qamJD}, we define quasi-arithmetic mixtures (Definition~\ref{def:qamix}), show their potential role in defining a generalization of Jensen-Shannon divergence~\cite{JS-2019}, and discusses the  underlying information geometry of parametric and non-parametric statistical models closed under the operation of taking quasi-arithmetic mixtures.
We propose a geometric generalization of the Jensen-Shannon divergence (Definition~\ref{eq:geojsd}) based on affine connections~\cite{IG-2016} in Section~\ref{sec:nablaJS} which recovers the ordinary Jensen-Shannon divergence and the geometric Jensen-Shannon divergence=\cite{JS-2019} when the affine connections are chosen as the mixture connection $\nabla^m$ and the exponential connection $\nabla^e$ of information geometry, respectively.

\section{Quasi-arithmetic averages and information geometry}\label{sec:qaaig}

\subsection{Scalar quasi-arithmetic means}\label{sec:QAM}

Let $\Delta_{n-1}=\{(w_1,\ldots,w_n) \st w_i\geq 0, \sum_{i=1}^n w_i=1\}\subset\bbR^n$ denotes the closed $(n-1)$-dimensional standard simplex sitting in $\bbR^n$ and 
$\Delta_{n-1}^\circ=\Delta_{n-1}\backslash\partial\Delta_{n-1}$ the open standard simplex where $\partial$ denotes the topological set boundary operator.
Weighted quasi-arithmetic means~\cite{Inequalities-1952} generalize the ordinary weighted arithmetic mean 
$A(x_1,\ldots,x_n;w)=\sum_i w_i x_i$ as follows:

\begin{definition}[Weighted quasi-arithmetic mean]\label{def:qam}
Let $f:I\subset\bbR\rightarrow\bbR$ be a strictly monotone and differentiable real-valued function. 
The weighted quasi-arithmetic mean (QAM) $m_f(x_1,\ldots,x_n;w)$ between $n$ scalars $x_1,\ldots,x_n\in  I\subset\bbR$ with respect to a normalized weight vector $w\in\Delta_{n-1}$, is defined by
$$
m_{f}(x_1,\ldots,x_n;w) := f^{-1}\left(\sum_{i=1}^n w_i f(x_i)\right).
$$
\end{definition}
The notion of quasi-arithmetic means and its properties were historically defined and studied independently by Knopp~\cite{Knopp-1928}, Jessen~~\cite{Jessen-1931},
Kolmogorov~\cite{Kolmogorov-1930}, Nagumo~\cite{Nagumo-1930} and De Finetti~\cite{de1931sul} in the  late 1920's-early 1930's 
(see also Acz{\'e}l~\cite{aczel1948mean}).
These quasi-arithmetic means are thus sometimes referred to in the literature Kolmogorov-Nagumo means~\cite{komori2021unified,czachor2002thermostatistics} or Kolmogorov-Nagumo-De Finetti means~\cite{bliman2022tiered}. 

Let us write for short $m_{f}(x_1,\ldots,x_n):=m_{f}(x_1,\ldots,x_n;\frac{1}{n},\ldots,\frac{1}{n})$ the quasi-arithmetic mean, and $m_{f,\alpha}(x,y):=m_f(x,y;\alpha,1-\alpha)$, the weighted bivariate mean. 
Mean $m_{f}(x_1,\ldots,x_n) $ is called a quasi-arithmetic mean because we have:
$$
f(m_f(x_1,\ldots,x_n))=\frac{1}{n} \sum_i f(x_i)=A(f(x_1),\ldots,f(x_n)),
$$ 
the arithmetic mean with respect to the $f$-representation~\cite{zhang2013nonparametric} of scalars.
A QAM has also been called a $f$-mean  in the literature (e.g.,~\cite{abramovich2005variant}) to emphasize its underlying generator $f$.
A QAM like any other generic mean~\cite{Bullen-2013} satisfies the {\em in-betweenness property}: 
$$
\min\{x_1,\ldots,x_n\}\leq m_{f}(x_1,\ldots,x_n;w)\leq \max\{x_1,\ldots,x_n\}.
$$
See also the recent works on aggregators~\cite{calvo2002aggregation}.
QAMs have been used in machine learning (e.g.,~\cite{komori2021unified}) and statistics (e.g.,~\cite{ballester2007using}).

We have the following invariance property of QAMs:

\begin{property}[Invariance of quasi-arithmetic mean~\cite{ConvexBook-2006}]\label{propr:invarqam}
 $m_{g}(x,y)=m_{f}(x,y)$ if and only if $g(t)=\lambda f(t)+c$ for $\lambda\in\bbR\backslash\{0\}$ and $c\in\bbR$.
\end{property}
See~\cite{bajraktarevic1958equation,losonczi1999equality} for the more general case of invariance of weighted quasi-arithmetic means with weights defined by functions.

Let $\CM(a,b)$ denotes the class of  continuous strictly monotone functions on $[a,b]$, and $\sim$ the equivalence relation $f\sim g$ 
if and only if $m_f=m_g$. Then the quasi-arithmetic mean induced by $f$ is $m_{[f]}$ where $[f]$ denotes the equivalence class of functions in $\CM(a,b)$ which contains $f$.
When $f(t)=t$, we recover the arithmetic mean $A$: $m_{\mathrm{id}}(x,y)=A(x,y)$ where $\mathrm{id}(x)=x$ is the scalar identity function.

The power means $m_p(x,y):=m_{f_p}(x,y)=\left(\frac{x^p+y^p}{2}\right)^{\frac{1}{p}}$, also called H\"older~\cite{pales2004hardy,zhang2009convexity} or sometimes Minkowski means~\cite{angulo2013morphological}, are 
obtained for the following {\em continuous family} of QAM generators $f_p(t)$ index by $p\in\bbR$: 
$$
f_p(t)=
\left\{
\begin{array}{ll}
\frac{t^p-1}{p}, & p\in\bbR\backslash\{0\},\cr
\log(t), & p=0.
\end{array}
\right., \quad
f_p^{-1}(t)=
\left\{
\begin{array}{ll}
(1+tp)^{\frac{1}{p}}, & p\in\bbR\backslash\{0\},\cr
\exp(t), & p=0.
\end{array}
\right.,
$$
Special cases of the power means are the harmonic mean ($H=m_{-1}$), the geometric mean ($G=m_{0}$), 
the arithmetic mean ($A=m_{1}$), and the quadratic mean ($Q=m_{2}$).

A QAM is said {\em positively homogeneous} if and only if $m_f(\lambda x,\lambda y)=\lambda\, m_f(x,y)$ for all $\lambda>0$.
The power means $m_p$ are provably the only positively homogeneous QAMs~\cite{Inequalities-1952}.

QAMs provide a versatile way to construct means~\cite{Bullen-2013} by specifying a functional generator $f\in\CM(I)$.
For example, the log-sum-exp mean\footnote{Also called the exponential mean~\cite{Bullen-2013} since it is a $f$-mean for the exponential function.} is obtained for the QAM generator $f_\LSE(t)=\exp(u)=f_0^{-1}(t)$ with $f^{-1}_\LSE(t)=\log u=f_0(t)$ (notice that these functions are the inverse of the geometric mean functions):
$$
\LSE(x,y)=\log\left(\frac{\exp^x+\exp^y}{2}\right)=m_{f_\LSE}(x,y).
$$ 

Quasi-arithmetic means have been generalized to complex-valued generators in~\cite{akaoka2022limit} and operators in~\cite{micic2011jensen}.

\subsection{Quasi-arithmetic averages}\label{sec:QAA}

To generalize scalar QAMs to other non-scalar types such as vectors or matrices, we have to face two difficulties:
\begin{enumerate}
\item First, we need to ensure that the generator $G:\bbX\rightarrow\bbR$ admits a continuously smooth global inverse $G^{-1}$, and 
\item Second, we would like the smooth function $G$ to bear a generalization of monotonicity of univariate functions.
\end{enumerate}

Indeed, the inverse function theorem~\cite{krantz2002implicit,InverseFunctionThm-1976}  in multivariable calculus states only  the existence {\em locally} of an inverse continuously differentiable function $G^{-1}$ for a multivariate function $G$ provided that the  Jacobian matrix of $G$ is not singular (i.e., Jacobian matrix has non-zero determinant).

We shall thus consider a well-behaved class $\calF$ of non-scalar  functions $G$ (i.e., vector or matrix functions) which admits global inverse functions $G^{-1}$ belonging to the same class $\calF$: Namely, we consider the gradient maps of Legendre-type functions where Legendre-type functions are defined as follows:

\begin{definition}[Legendre type function~\cite{Rockafellar-1967}]\label{def:Legendre}
$(\Theta,F)$ is of Legendre type if the function $F:\Theta\subset\bbX\rightarrow\bbR$ is strictly convex and differentiable with $\Theta\not=\emptyset$ and
\begin{equation}\label{eq:cond}
\lim_{\lambda\rightarrow 0} \frac{d}{\dlambda} F(\lambda\theta+(1-\lambda)\bar\theta)=-\infty,\quad \forall\theta\in\Theta, \forall\bar\theta\in\partial\Theta.
\end{equation}
\end{definition}
The condition of Eq.~\ref{eq:cond} is related to the notion of steepness in exponential families~\cite{BN-2014}.

Legendre-type functions $F(\Theta)$ admits a convex conjugate $F^*(\eta)$ via the
Legendre transform
$$
F^*(\eta)=\inner{\nabla F^{-1}(\eta)}{\eta}-F(\nabla F^{-1}(\eta)),
$$
where $\inner{\theta}{\eta}=\theta^\top\eta$ denotes the inner product in $\bbX$ (e.g., Euclidean inner product $\inner{\theta}{\eta}=\theta^\top\eta$ for $\bbX=\bbR^d$, the Hilbert-Schmidt inner product
$\inner{A}{B}:=\tr(AB^\top)$ where $\tr(\cdot)$ denotes the matrix trace for $\bbX=\mathrm{Mat}_{d,d}(\bbR)$, etc.), and
 $\eta\in H$  with $H$ the image of the gradient map $\nabla F:\Theta\rightarrow H$.
Convex conjugate $F^*(\eta)$ is of Legendre type (Theorem 1~\cite{Rockafellar-1967}).
Moreover, we have $\nabla F^*=\nabla F^{-1}$.

The gradient of a strictly convex function of Legendre type can also be interpreted as a  generalization
 the notion of monotonicity of a univariate function:
  A function $G:\bbX\rightarrow \bbR$ is said {\em strictly increasing co-monotone} if
$$
\forall \theta_1,\theta_2\in\bbX, \theta_1\not=\theta_2,\quad \inner{\theta_1-\theta_2}{G(\theta_1)-G(\theta_2)}>0.
$$
and strictly decreasing co-monotone if $-G$ is strictly increasing co-monotone.

\begin{proposition}[Gradient co-monotonicity]\label{prop:gradcomonotone}
The gradient functions $\nabla F(\theta)$ and $\nabla F^*(\eta)$ of the Legendre-type convex conjugates $F$ and $F^*$ in $\calF$ are strictly increasing co-monotone functions.
\end{proposition}

\begin{proof}
We have to prove that
\begin{eqnarray}
\inner{\theta_2-\theta_1}{\nabla F(\theta_2)-\nabla F(\theta_1)}&>&0,\quad  \forall \theta_1\not=\theta_2 \in\Theta\label{eq:co1}\\
\inner{\eta_2-\eta_1}{\nabla F^*(\eta_2)-\nabla F^*(\eta_1)}&>&0, \quad \forall \eta_1\not=\eta_2\in H \label{eq:co2}
\end{eqnarray}

The inequalities follow by interpreting the terms of the left-hand-side of Eq.~\ref{eq:co1} and Eq.~\ref{eq:co2} as Jeffreys-symmetrization~\cite{JS-2019} of the dual Bregman divergences~\cite{Bregman-1967}:
\begin{eqnarray*}
B_F(\theta_1:\theta_2)&=& F(\theta_1)-F(\theta_2)-\inner{\theta_1-\theta_2}{\nabla F(\theta_2)}\geq 0,\\
B_{F^*}(\eta_1:\eta_2)&=& F^*(\eta_1)-F^*(\eta_2)-\inner{\eta_1-\eta_2}{\nabla F(\theta_2)}\geq 0,
\end{eqnarray*}
where the first equality holds if and only if $\theta_1=\theta_2$ and the second inequality holds iff $\eta_1=\eta_2$.
Indeed, we have the following Jeffreys-symmetrization of the dual Bregman divergences $B_{F}$ and $B_{F^*}$: 
\begin{eqnarray*}
B_{F}(\theta_1:\theta_2)+B_{F}(\theta_2:\theta_1)&=& \inner{\theta_2-\theta_1}{\nabla F(\theta_2)-\nabla F(\theta_1)}>0, \quad \forall \theta_1\not=\theta_2\\
B_{F^*}(\eta_1:\eta_2)+B_{F^*}(\eta_2:\eta_1)&=& \inner{\eta_2-\eta_1}{\nabla F^*(\eta_2)-\nabla F^*(\eta_1)}>0, \quad \forall \eta_1\not=\eta_2\\
\end{eqnarray*}

The symmetric divergences $\mathrm{JB}_F(\theta_1,\theta_2):=B_{F}(\theta_1:\theta_2)+B_{F}(\theta_2:\theta_1)$ and
$\mathrm{JB}_{F^*}(\eta_1,\eta_2):=B_{F^*}(\eta_1:\eta_2)+B_{F^*}(\eta_2:\eta_1)$ are called Jeffreys-Bregman divergences in~\cite{BR-2011}.
\end{proof}

\begin{remark}
Co-monotonicity can be interpreted as a multivariate generalization of monotone univariate functions:
A smooth univariate strictly increasing monotone function $f$ is such that $f'(x)>0$. 
Since $f'(x)=\lim_{h\rightarrow 0} \frac{f(x+h)-f(x)}{h}$, a strictly monotone function is such that $(x+h-x)\, (f(x+h)-f(x))>0$ for small enough $h>0$.
\end{remark}

Let us now define the weighted quasi-arithmetic averages (QAAs) as  follows:

\begin{definition}[Weighted quasi-arithmetic averages]\label{def:mqam}
Let $F:\Theta\rightarrow \bbR$ be a strictly convex and smooth real-valued function of Legendre-type in $\calF$.
The weighted quasi-arithmetic average of $\theta_1,\ldots,\theta_n$ and $w\in\Delta_{n-1}$ is defined by the gradient map $\nabla F$ as follows:
\begin{eqnarray}
M_{\nabla F}(\theta_1,\ldots,\theta_n;w) &:=&  
  {\nabla F}^*\left(\sum_{i=1}^n w_i \nabla F(\theta_i) \right),\\
&=& {\nabla F}^{-1}\left(\sum_{i=1}^n w_i \nabla F(\theta_i)\right),
\end{eqnarray}
where ${\nabla F}^*={\nabla F}^{-1}$ is the gradient map of the Legendre transform $F^*$ of $F$.
\end{definition}

We recover the usual definition of scalar QAMs $m_{f}$ (Definition~\ref{def:qam}) 
when $F(t)=\int^t_a f(u)\du$ for a strictly increasing or strictly decreasing and continuous function $f$: 
$m_f=M_{F'}$ (with $f^{-1}=(F')^{-1}$). 
Notice that we only need to consider $F$ to be strictly convex or strictly concave and smooth to define a multivariate QAM since $M_{\nabla F}=M_{-\nabla F}$.

The quasi-arithmetic averages can also be called $\nabla F$-means since we have
$$
\nabla F\left(M_{\nabla F}(\theta_1,\ldots,\theta_n;w)\right)=\sum_{i=1}^n w_i \nabla F(\theta_i)=A(\nabla F(\theta_1),\ldots, \nabla F(\theta_n);w),
$$
the ordinary weighted arithmetic mean on the $\nabla F$-representations.

Let us give some examples of vector and matrix quasi-arithmetic averages:

\begin{example}[Separable quasi-arithmetic average]\label{ex:vecex}
When the strictly convex $d$-variate real-valued function $F(\theta)$ is separable, i.e., $F(\theta)=\sum_{i=1}^d f_i(\theta_i)$ with $f_i:I_i\rightarrow\bbR$ for strictly convex and differentiable univariate functions $f_i(\theta_i)\in\mathcal{CM}(I_i)$, the global gradient maps are 
$\nabla F(\theta)=\vectorcol{f_1'(\theta_1)}{\vdots}{f_d'(\theta_d)}$ and 
$\nabla F^*(\eta)=\vectorcol{f_1'^{-1}(\eta_1)}{\vdots}{f_d'^{-1}(\eta_d)}=\nabla F^{-1}(\eta)$ so that 
 we have $M_{\nabla F}(\theta,\theta')=\vectorcol{M_{f_1'}(\theta_1,\theta_1')}{\vdots}{M_{f_d'}(\theta_d,\theta_d')}$, the componentwise quasi-arithmetic scalar means.
\end{example}

\begin{example}[Non-separable quasi-arithmetic average]\label{ex:categorical}
Consider the non-separable  $d$-variate real-valued function $F(\theta)=\log(1+\exp_{i=1}^d e^{\theta_i})=\LSE(0,\theta_1,\ldots,\theta_2)$.
This function called $\LSE_0^+(\theta_1,\ldots,\theta_2)=\LSE(0,\theta_1,\ldots,\theta_2)$ is strictly convex and differentiable of Legendre type~\cite{MCIG-2019}, with the reciprocal gradient maps
$\nabla F(\theta)=\vectorcol{\frac{e^{\theta_1}}{1+\sum_{j=1}^d e^{\theta_j}}}{\vdots}{\frac{e^{\theta_d}}{1+\sum_{j=1}^d e^{\theta_j}}}$ and
$\nabla F^*(\eta)=\vectorcol{\log\frac{\eta_1}{1-\sum_{j=1}^d \eta_j}}{\vdots}{\log\frac{\eta_d}{1-\sum_{j=1}^d \eta_j}}$.
We shall call this quasi-arithmetic average the categorical mean as it is induced by the cumulant function of  the family of categorical distributions 
(see \S\ref{sec:qaaigd}).
\end{example}

\begin{example}[Matrix example]\label{ex:logdet}
Consider the strictly convex function~\cite{webster1994convexity,boyd2004convex}: 
\begin{eqnarray*}
F&:& \SPD \rightarrow \bbR\\
\theta &\mapsto& -\log\det(\theta), 
\end{eqnarray*}
where $\det(\cdot)$ denotes the matrix determinant.
Function $F(\theta)$ is strictly convex and differentiable~\cite{boyd2004convex} on the domain of $d$-dimensional symmetric positive-definite matrices $\SPD$ (open convex cone).  We have
$F(\theta) = -\log\det(\theta)$,
$\nabla F(\theta) = -\theta^{-1} =:\eta(\theta)$,
$\nabla F^{-1}(\eta) = -\eta^{-1}=:\theta(\eta)$, and
$F^*(\eta) = \inner{\theta(\eta)}{\eta}-F(\theta(\eta))= -d-\log\det(-\eta)$,
where the dual parameter $\eta$ belongs to the $d$-dimensional negative-definite matrix domain, and the inner matrix product is the Hilbert-Schmidt inner product
$\inner{A}{B}:=\tr(AB^\top)$, where $\tr(\cdot)$ denotes the matrix trace.
It follows that 
$M_{\nabla F}(\theta_1,\theta_2)=2(\theta_1^{-1}+\theta_2^{-1})^{-1}$ 
is the matrix harmonic mean~\cite{matrixharmonic-1997} generalizing the scalar harmonic mean $H(a,b)=\frac{2ab}{a+b}$ for $a,b>0$.
Notice that the quasi-arithmetic center with respect to  $F^*$ is 
$M_{\nabla F^*}(\eta_1,\eta_2)=2\, \left(\eta_1^{-1}+\eta_2^{-1}\right)^{-1}$.
Thus in that case, we have $M_{\nabla F}=M_{\nabla F^*}$. 
That is, the gradient maps of convex conjugates yield the same quasi-arithmetic average
Other examples of matrix means are reported in~\cite{MatHellingerBhatia-2019}.
\end{example}

\section{Use of quasi-arithmetic averages in dually flat manifolds}\label{sec:qaaigd}

In this section, we shall elicit the roles of quasi-arithmetic averages in information geometry~\cite{IG-2016}, and report the invariance and equivariance  properties of quasi-arithmetic averages with respect to affine transformations from the lens of information geometry.

Let $(M,g,\nabla,\nabla^*)$ be a dually flat space (DFS) where $\nabla$ and $\nabla^*$ are the dual torsion-free flat affine connections such that $\frac{\nabla+\nabla^*}{2}$ is the Riemannian metric Levi-Civita connection $\nabla^g$ induced by $g$ (we have $\nabla^*=2\nabla^g-\nabla$ and ${\nabla^*}^*=\nabla$).
Let $F(\theta)$ and $F^*(\eta)$ denotes the Legendre-type potential functions with $\theta$ denoting the $\nabla$-affine coordinate system and $\eta$ denoting the $\nabla^*$-affine coordinate system. 
A point $P$ in a DFS can thus be represented  either by the coordinates $\theta(P)$ or by the coordinates $\eta(P)$.
Let us denote this duality of coordinates by $P\,\vectortwo{\theta(P)}{\eta(P)}$.
In a DFS, the dual canonical divergences~\cite{IG-2016} $D_{\nabla,\nabla^*}(P:Q)$ and $D_{\nabla,\nabla^*}^*(P:Q)=D_{\nabla^*,\nabla}(P:Q)$ between two points $P$ and $Q$ of $M$ can be expressed using the coordinate systems as dual Bregman divergences. 
We have the following identities: 
$$
D_{\nabla,\nabla^*}(P:Q)=B_F(\theta(P):\theta(Q))=B_{F^*}(\eta(Q):\eta(P))=D_{\nabla^*,\nabla}(Q:P).
$$

\subsection{Quasi-arithmetic averages in dual parameterizations of dual geodesics}\label{sec:geo}

\begin{figure}[hbtp]
\centering

\includegraphics[width=0.7\textwidth]{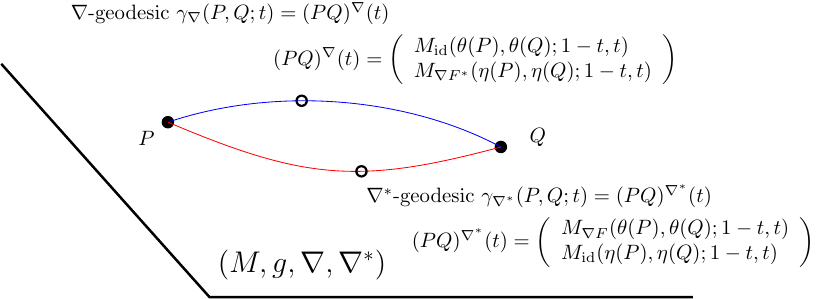}

\caption{The points on dual geodesics in a dually flat spaces have dual coordinates expressed with quasi-arithmetic averages.\label{fig:igdualgeoqams}}
\end{figure}

In a DFS $(M,g,\nabla,\nabla^*)=\DFS(F,\theta\in\Theta;F^*,\eta\in H)$, the primal geodesics $\gamma_\nabla(P,Q;t)$ are obtained as line segments in the $\theta$-coordinate system (because the Christoffel symbols of the connection $\nabla$ vanishes in the $\theta$-coordinate system) while the dual geodesics $\gamma_{\nabla^*}(P,Q;t)$ are line segments in the $\eta$-coordinate system (because the Christoffel symbols of the dual connection $\nabla^*$ vanishes in the $\eta$-coordinate system). 
The dual geodesics define interpolation schemes $(PQ)^\nabla(t)=\gamma_\nabla(P,Q;t)$ and $(PQ)^{\nabla^*}(t)=\gamma_{\nabla^*}(P,Q;t)$  between input points $P$ and $Q$ with 
$P=\gamma_\nabla(P,Q;0)=\gamma_{\nabla^*}(P,Q;0)$ and $Q=\gamma_\nabla(P,Q;1)=\gamma_{\nabla^*}(P,Q;1)$ when $t$ ranges in $[0,1]$.
We express the coordinates of the interpolated points on $\gamma_\nabla$ and $\gamma_{\nabla^*}$  using quasi-arithmetic averages as follows (Figure~\ref{fig:igdualgeoqams}):

\begin{eqnarray}
(PQ)^\nabla(t)&=&\gamma_\nabla(P,Q;t)=\vectortwo{M_\idfunc(\theta(P),\theta(Q);1-t,t)}{M_{\nabla F^*}(\eta(P),\eta(Q);1-t,t)},\\
(PQ)^{\nabla^*}(t)&=&\gamma_{\nabla^*}(P,Q;t)=\vectortwo{M_{\nabla F}(\theta(P),\theta(Q);1-t,t)}{M_\idfunc(\eta(P),\eta(Q);1-t,t)}.
\end{eqnarray}
Quasi-arithmetic averages were used by a geodesic bisection algorithm to approximate the circumcenter of the minimum enclosing balls with respect to the canonical divergence in a DFS in~\cite{nock2005fitting}.

\subsection{Quasi-arithmetic average coordinates of dual barycenters with respect to the canonical divergence}\label{sec:bary}
Consider a finite set of $n$ points $P_1,\ldots, P_n$ on the DFS $(M,g,\nabla,\nabla^*)$.
The points $P_i\,\vectortwo{\theta_i}{\eta_i}$ can be expressed in the dual coordinate systems either as $\theta(P_i)=\theta_i$ or $\eta(P_i)=\eta_i$.
The right centroid point $\bar C_R\in M$ defined by $\bar C_R=\arg\min_{P\in M} \sum_{i=1}^n \frac{1}{n} D_{\nabla,\nabla^*}(P_i:P)$ (or equivalently as a right-sided Bregman centroid~\cite{banerjee2005clustering} $\bar\theta_R=\arg\min_{\theta} \sum_{i=1}^n \frac{1}{n} B_F(\theta_i:\theta)$)  has dual coordinates
\begin{eqnarray}
\bar\theta_R &=& \theta(\bar C_R)  = \frac{1}{n}\sum_{i=1}^n \theta_i =M_{\idfunc}(\theta_1,\ldots,\theta_n),\\
\bar\eta_R &=& \nabla F(\bar\theta_R) = M_{\nabla F^*}(\eta_1,\ldots,\eta_n).
\end{eqnarray}
Similarly, the left centroid point 
$\bar C_L\in M$ defined by $\bar C_L=\arg\min_{P\in M} \sum_{i=1}^n \frac{1}{n} D_{\nabla,\nabla^*}(P:P_i)$ (or equivalently a left-sided Bregman centroid~\cite{nielsen2009sided} $\bar\theta_L=\arg\min_{\theta} \sum_{i=1}^n \frac{1}{n} B_F(\theta:\theta_i)$) has coordinates
\begin{eqnarray}
\bar\theta_L &=& M_{\nabla F}(\theta_1,\ldots,\theta_n),\\
\bar\eta_L &=& \nabla F(\bar\theta_L) =  M_{\idfunc}(\eta_1,\ldots,\eta_n).
\end{eqnarray}
Thus we have the two   dual sided centroids $\bar C_R$ and $\bar C_L$ (reference duality~\cite{zhang2013nonparametric}) on the dually flat manifold $M$ expressed using the dual coordinates as
$$
\bar C_R\,
\vectortwo{\bar\theta_L = M_{\nabla F}(\theta_1,\ldots,\theta_n)}{\bar\eta_L = \nabla F(\bar\theta_L) =  M_{\idfunc}(\eta_1,\ldots,\eta_n)},\quad
\bar C_L\,
\vectortwo{\bar\theta_L = M_{\nabla F}(\theta_1,\ldots,\theta_n)}{\bar\eta_L = \nabla F(\bar\theta_L) =  M_{\idfunc}(\eta_1,\ldots,\eta_n)}
$$

Let $\underline{\theta}=\frac{1}{n}\sum_{i=1}^n \theta_i$ and $\underline{\eta}=\frac{1}{n}\sum_{i=1}^n\eta_i$.
Then we have $\underline{\theta}=\nabla F^*(M_{\nabla F^*}(\eta_1,\ldots,\eta_n))$ and 
$\underline{\eta}=\nabla F(M_{\nabla F}(\theta_1,\ldots,\theta_n))$ .
The dual DFS centroids~\cite{pelletier2005informative} were studied as Bregman sided centroids and expressed as quasi-arithmetic averages in~\cite{nielsen2009sided}.

\begin{figure}
\centering

\includegraphics[width=0.7\textwidth]{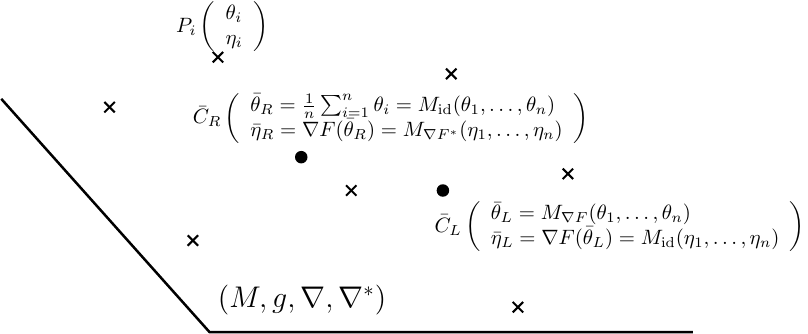}

\caption{Dual centroids in a dually flat spaces have dual coordinates expressed with  quasi-arithmetic averages.\label{fig:igdualqams}}
\end{figure}

Figure~\ref{fig:igdualqams} illustrates the dual centroids expressed using 
the quasi-arithmetic averages.

We may consider the barycenter of $n$ weighted points $P_1,\ldots, P_n$ (weight vector $w\in\Delta_{n-1}$) with respect to a Jensen divergence~\cite{BR-2011} $J_F$ defined as the minimization of $\sum_{i=1}^n w_i J_F(\theta,\theta_i)$.
In~\cite{BR-2011}, the following iterative algorithm was proposed: Let $\theta^{(0)}= \sum_{i=1}^n w_i\theta_i$, and iteratively update 
$\theta^{(t+1)}=M_{\nabla F}\left(\frac{\theta^{(t)}+\theta_i}{2},\ldots,\frac{\theta^{(t)}+\theta_n}{2};w\right)$.

\subsection{Tripartite duality of densities/divergences/means and dual quasi-arithmetic averages}\label{sec:dualQAA}
Banerjee et al.~\cite{banerjee2005clustering} proved a bijection between natural regular exponential families 
$\calE_F=\{p_\theta(x)=\exp(x\cdot\theta-F(\theta))\}$ with cumulant functions $F$ and ``regular'' dual Bregman divergences $B_{F^*}$ by rewriting the densities as $p_\theta(x)=p_\eta(x)=\exp(-B_{F^*}(x:\eta)+F^*(x))$ with $\eta=\nabla F(\theta)$  (using the Young equality $F(\theta)+F^*(\eta)-\inner{\theta}{\eta}=0$). 
Furthermore, a bijection between Bregman divergences $B_F$ and quasi-arithmetic averages  $M_{\nabla F}$ was informally mentioned in~\cite{nock2005fitting}. Using these bijections, we can cast the maximum likelihood estimator (MLE) of an exponential family $\calE_F$ as a dual right-sided Bregman centroid problem~\cite{nielsen2012k}.
Figure~\ref{fig:dualqams} summarizes the dualities between convex conjugate functions, exponential families and Bregman divergences, and maximum likelihood estimator, Bregman centroid expressed as multivariate QAMs.

The categorical mean of Example~\ref{ex:categorical} is induced by the gradient map of the cumulant function of the exponential family of categorical distributions~\cite{IG-2016} (the family of discrete distributions on a finite alphabet). 

\begin{figure}
\centering

\includegraphics[width=\textwidth]{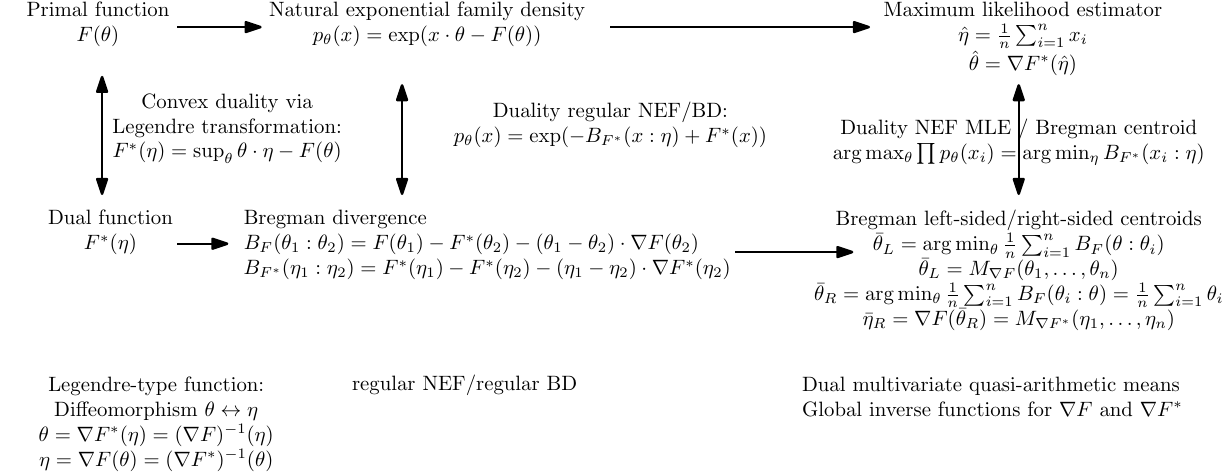}

\caption{Overview of the bijections between regular exponential families, Bregman divergences of Legendre-type, and quasi-arithmetic averages.\label{fig:dualqams}}
\end{figure}

A Legendre-type function $F$ induces two dual quasi-arithmetic weighted averages $M_{\nabla F}$ and $M_{\nabla F^*}$ by the gradient maps of the convex conjugates $F$ and $F^*$.

When $\nabla F=\nabla F^*=\nabla F^{-1}$ (meaning that the convex conjugate gradients are reciprocal to each others), we have $M_{\nabla F}=M_{\nabla F^*}$. This holds for example for the scalar and matrix harmonic means which are self-dual means with $\nabla F(x)=x^{-1}=\nabla F^*(x)$.

Consider the Mahalanobis divergence $\Delta^2$ (i.e., the squared Mahalanobis distance $\Delta$) as a Bregman divergence obtained for the  quadratic form generator $F_Q(\theta)=\frac{1}{2}\theta^\top Q\theta+c\theta+\kappa$ for a symmetric positive-definite $d\times d$ matrix $Q$, $c\in\bbR^d$ and $\kappa\in\bbR$.
We have:
$$
\Delta^2(\theta_1,\theta_2)=B_{F_Q}(\theta_1:\theta_2)=
\frac{1}{2}(\theta_2-\theta_1)^\top Q (\theta_2-\theta_1).
$$
When $Q=I$, the identity matrix, the Mahalanobis divergence coincides with the Euclidean divergence\footnote{The squared Euclidean/Mahalanobis divergence are not metric distances since they fail the triangle inequality.} (i.e., the squared Euclidean distance).
The Legendre convex conjugate is $F^*(\eta)=\frac{1}{2}\eta^\top Q^{-1}\eta=F_{Q^{-1}}(\eta)$, 
and we have $\eta=\nabla F_Q(\theta)=Q\theta$ and $\theta=\nabla F^*_Q(\eta)=Q^{-1}\eta$. 
Thus we get the following dual quasi-arithmetic averages:
\begin{eqnarray}
M_{\nabla F_Q}(\theta_1,\ldots,\theta_n;w) &=&
Q^{-1}\left(\sum_{i=1}^n w_i Q\theta_i\right)=
\sum_{i=1}^n w_i \theta_i=M_{\idfunc}(\theta_1,\ldots,\theta_n;w),\\
M_{\nabla F_Q^*}(\eta_1,\ldots,\eta_n;w)&=&
Q\left(\sum_{i=1}^n w_i Q^{-1}\eta_i\right)=
M_{\idfunc}(\eta_1,\ldots,\eta_n;w).
\end{eqnarray}

The dual quasi-arithmetic average functions $M_{\nabla F_Q}$ and $M_{\nabla F_Q^*}$ induced by a Mahalanobis Bregman generator $F_Q$ coincide since 
$M_{\nabla F_Q}=M_{\nabla F_Q^*}=M_{\idfunc}$. This means geometrically that the left-sided and right-sided centroids of the underlying canonical divergences match. The average $M_{\nabla F_Q}(\theta_1,\ldots,\theta_n;w)$ expresses the centroid $C=\bar{C}_R=\bar{C}_L$ in the $\theta$-coordinate system ($\theta(C)=\underline{\theta}$) 
and the average $M_{\nabla F_Q^*}(\eta_1,\ldots,\eta_n;w)$ expresses the same centroid in the $\eta$-coordinate system
 ($\eta(C)=\underline{\eta}$). In that case of self-dual flat Euclidean geometry, there is an affine transformation relating
 the $\theta$- and $\eta$-coordinate systems:$\eta=Q\theta$ and $\theta=Q^{-1}\eta$.
As we shall see this is because the underlying geometry is self-dual Euclidean flat space $(M,g_{\mathrm{Euclidean}},\nabla_{\mathrm{Euclidean}},\nabla_{\mathrm{Euclidean}}^*=\nabla_{\mathrm{Euclidean}})$ and that both dual connections coincide with the Euclidean connection (i.e., the Levi-Civita connection of the Euclidean metric). In this particular case, the dual coordinate systems are just related by affine transformations of one to another.

\subsection{Quasi-arithmetic averages and the inductive matrix geometric mean}\label{sec:inductive}

Consider $P$ and $Q$ two symmetric positive-definite (SPD) matrices of $\Sym_{++}(d)$.
By equipping the SPD cone $\Sym_{++}(d)$ with the Riemannian trace metric 
$$
g_P(X,Y)=\tr\left(P^{-1}X\, P^{-1}Y\right)
$$ 
where $X$ and $Y$ are symmetric matrices of the tangent plane $T_p$ identified with the vector space $\Sym(d)$, we
 get a Riemannian manifold $(\Sym_{++}(d),g)$ with geodesic distance~\cite{james1973variance}:
$$
\rho(P,Q)=\sqrt{  \sum_{i=1}^d \log^2 \lambda_i(P^{-\frac{1}{2}}\, Q\, P^{-\frac{1}{2}})},
$$
where $\lambda_i$ denote the $i$-th largest real-valued eigenvalue of the SPD matrix $P^{-\frac{1}{2}}\, Q\, P^{-\frac{1}{2}}$.
The Riemannian center of mass $P^*$ of $n$ points $P_1$, \ldots, $P_n$ (commonly called centroid or K\"archer mean) is defined as
$$
P^* = \arg\min_{P\in \Sym_{++}(d)} \frac{1}{n} \sum_{i=1}^n \rho^2(P_i,P).
$$
Since the SPD Riemannian manifold $(\Sym_{++}(d),g)$ is of non-positive sectional curvatures ranging in $[-\frac{1}{2},0]$, the Riemannian centroid $P^*$ is unique.
In particular, when $n=2$, we get
$$
P^* = P_1^{\frac{1}{2}}\, \left(P_1^{-\frac{1}{2}}\, P_2\, P_1^{-\frac{1}{2}}\right)^{\frac{1}{2}}\, P_1^{\frac{1}{2}}
$$
which coincides with one usual definition~\cite{ando2004geometric,bhatia2006riemannian} of the geometric matrix mean $G(P_1,P_2)$ 
where 
$$
G(P,Q)=Q^{\frac{1}{2}}  \left(Q^{-\frac{1}{2}}\, P\, Q^{-\frac{1}{2}}\right)^{\frac{1}{2}}  Q^{\frac{1}{2}},
$$

Nakamura~\cite{AHM-Nakamura-2001} considered the following iterations based on the arithmetic matrix mean $A(P,Q)=(P+Q)/2$ and harmonic matrix mean
 $H(P,Q)=2\, ((P^{-1}+Q^{-1}))^{-1}$:
\begin{eqnarray}
P_{t+1} &=& \frac{P_t+Q_t}{2}=:A(P_t,Q_t),\label{eq:ahm1}\\ 
Q_{t+1} &=& 2\, (P_t^{-1}+Q_t^{-1})^{-1}=:H(P_t,Q_t),\label{eq:ahm2}
\end{eqnarray}
initialized with $P_0=P$ and $Q_0=Q$.
Let $M(P,Q)=\lim_{t\rightarrow \infty} P_t$.
It is proven that $M(P,Q)=\lim_{t\rightarrow \infty} Q_t$, and
$$
M(P,I)=P^\frac{1}{2},
$$ 
the square-root matrix,
and
$$
M(P,Q)=G(P,Q).
$$
Furthermore, the convergence is quadratic~\cite{AHM-Nakamura-2001,atteia2001self}.

We can extend $(\Sym_++(d),g)$ as a dually flat space $(\Sym_++(d),g,\nabla,\nabla^*)$ where $\nabla$ is the flat Levi-Civita connection induced by the Euclidean metric 
$\labelg{E}_P(X,Y)=\tr(XY)$  and $\nabla^*$ is the flat Levi-Civita connection induced by the so-called inverse Euclidean metric~\cite{DFSAllMetricConnections-2019,geomixedEuclideanSPD-2022}
$\labelg{\IE}_P(X,Y)=\tr\left(P^{-2} X P^{-2} Y\right)$ (isometric to the Euclidean metric).
The non-flat trace metric $g$ is interpreted as a balanced bilinear form~\cite{DFSAllMetricConnections-2019}.
The midpoint $\nabla$-geodesic corresponds to the arithmetic mean and the midpoint $\nabla^*$-geodesic corresponds to the matrix harmonic mean.
The iterations of Eq.~\ref{eq:ahm1} and Eq.~\ref{eq:ahm2} converging to the Riemannian center of mass can thus be interpreted geometrically on the dually flat space $(\Sym_++(d),g,\nabla,\nabla^*)$ (see Figure~\ref{fig:inductivemean}),
with the geodesic midpoints  expressed as quasi-arithmetic averages $M_X$ and $M_{X^-1}$ which are the gradient maps of Legendre-type functions 
$\frac{1}{2}\tr(X^2)$ and $-\log\det(X)$, respectively.

\begin{figure}[hbtp]
\centering

\includegraphics[width=0.7\textwidth]{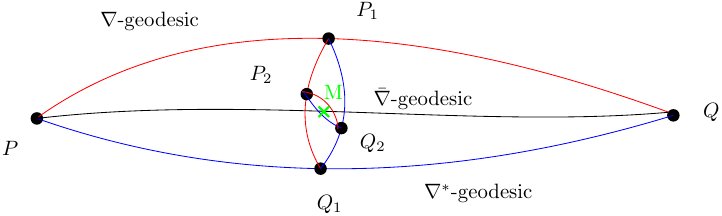}

\caption{The points on dual geodesics in a dually flat spaces have dual coordinates expressed with quasi-arithmetic averages.\label{fig:inductivemean}}
\end{figure}

The inductive process is further generalized to Hilbert spaces of functions with the arithmetic and harmonic matrix means being replaced by the arithmetic average function $A(p,q)=\frac{p+q}{2}$ and a harmonic-type function $hH(p,q)=\left(\frac{p^*+q^*}{2}\right)^*$ defined using the Legendre transform in~\cite{atteia2001self}.

\section{Invariance and equivariance properties of quasi-arithmetic averages}\label{sec:invariance}
Recall that a dually flat manifold~\cite{IG-2016} $(M,g,\nabla,\nabla^*)$ has a canonical divergence~\cite{amari1985differential} $D_{\nabla,\nabla^*}$ which can be expressed either as a primal Bregman divergence in the $\nabla$-affine coordinate system $\theta$ (using the  convex potential function $F(\theta)$)
 or as a dual Bregman divergence   in the $\nabla^*$-affine coordinate system $\eta$  (using the  convex conjugate potential function $F^*(\eta)$), or as dual Fenchel-Young divergences~\cite{nielsen2022statistical} using the mixed coordinate systems $\theta$ and $\eta$. 
The dually flat manifold $(M,g,\nabla,\nabla^*)$ (a particular case of Hessian manifolds~\cite{shima1997geometry} which admit a global coordinate system) is characterized by $(\theta,F(\theta);\eta,F^*(\eta))$ which we shall denote by
 $(M,g,\nabla,\nabla^*)\leftarrow \DFS(\theta,F(\theta);\eta,F^*(\eta))$ (or in short $(M,g,\nabla,\nabla^*)\leftarrow (\Theta,F(\theta))$). 
However, the choices of parameters $\theta$ and $\eta$ and potential functions $F$ and $F^*$ are not unique since they can be chosen up to affine reparameterizations and additive affine terms~\cite{IG-2016}: 
$(M,g,\nabla,\nabla^*)\leftarrow \DFS([\theta,F(\theta);\eta,F^*(\eta)])$ where $[\cdot]$ denotes the equivalence class  that has been called purposely the affine Legendre invariance in~\cite{nakajima2021dually} (see Section~\ref{sec:leginvar}):

\begin{itemize}
\item First, consider changing the potential function $F(\theta)$ by adding an affine term:
$\barF(\theta)=F(\theta)+\inner{c}{\theta}+d$. 
We have $\nabla \barF(\theta)=\nabla F(\theta)+c=\etabar$.
Inverting $\nabla \barF(x)=\nabla F(x)+c=y$, we get $\nabla \barF^{-1}(y)=\nabla F(y-c)$.
We check that $B_F(\theta_1:\theta_2)=B_{\bar F}(\theta_1:\theta_2)=D_{\nabla,\nabla^*}(P_1:P_2)$ with $\theta(P_1)=:\theta_1$ and 
$\theta(P_2)=:\theta_2$. It is indeed well-known that Bregman divergences modulo affine terms coincide~\cite{banerjee2005clustering}.
For the quasi-arithmetic averages $M_{\nabla\barF}$ and $M_{\nabla F}$, we thus obtain the following invariance property:
$M_{\nabla \barF}(\theta_1,\ldots;\theta_n;w)=M_{\nabla F}(\theta_1,\ldots;\theta_n;w)$.

\item Second, consider an affine change of coordinates $\bartheta=A\theta+b$ for $A\in\GL(d)$ and $b\in\bbR^d$,
and define the potential function $\barF(\bartheta)$ such that $\barF(\bartheta)=F(\theta)$. 
We have  $\theta=A^{-1}(\bartheta-b)$ and $\barF(x)=F(A^{-1}(x-b))$.
It follows that $\nabla \barF(x)=(A^{-1})^\top \nabla F(A^{-1}(x-b))$, and we check that $B_{\barF(\overline{\theta_1}:\overline{\theta_2})}=B_F(\theta_1:\theta_2)$:
\begin{eqnarray*}
B_{\barF(\overline{\theta_1}:\overline{\theta_2})}&=&\barF(\overline{\theta_1})+\barF(\overline{\theta_2})
-\inner{\overline{\theta_1}-\overline{\theta_2}}{\nabla \barF(\overline{\theta_2})},\\
&=& F(\theta_1)-F(\theta_2)-(A(\theta_1-\theta_2))^\top
(A^{-1})^\top  \nabla F(\theta_2),\\
&=& F(\theta_1)-F(\theta_2)-(\theta_1-\theta_2)^\top \underbrace{A^\top  (A^{-1})^\top}_{(A^{-1}A)^\top=I} \nabla F(\theta_2)= B_F(\theta_1:\theta_2).
\end{eqnarray*}
This highlights the invariance that $D_{\nabla,\nabla^*}(P_1:P_2)=B_F(\theta_1:\theta_2)=B_{\barF(\bartheta_1:\bartheta_2)}$, i.e., the canonical divergence does not change under a reparameterization of the $\nabla$-affine coordinate system.
For the induced quasi-arithmetic averages $M_{\nabla\barF}$ and $M_{\nabla F}$, 
we have $\nabla \barF(x)=(A^{-1})^\top \nabla F(A^{-1}(x-b))=y$, we calculate 
$x=\nabla \barF(x)^{1}(y)=A\,\nabla \barF^{-1}( ((A^{-1})^\top)^{-1} y)+b$, and we have
\begin{eqnarray*}
M_{\nabla\barF}(\overline{\theta_1},\ldots,\overline{\theta_n};w)&:=& \nabla\barF^{-1}(\sum_i w_i \nabla\barF(\overline{\theta_i})),\\
&=& (\nabla\barF)^{-1}\left( (A^{-1})^\top \sum_i w_i \nabla F(\theta_i)\right),\\
&=& A\, \nabla F^{-1}\left( \underbrace{((A^{-1})^\top)^{-1} (A^{-1})^\top}_{=I}  \sum_i w_i \nabla F(\theta_i) \right),\\
M_{\nabla\barF}(\overline{\theta_1},\ldots,\overline{\theta_n};w) &=&  A\, M_{\nabla F}(\theta_1,\ldots,\theta_n;w)+b
\end{eqnarray*}

More generally, we may define $\bar F(\bartheta)=F(A\theta+b)+\inner{c}{\theta}+d$ and get via Legendre transformation  
 $\bar F^*(\bareta)=F^*(A^*\eta+b^*)+\inner{c^*}{\eta}+d^*$ (with $A^*,b^*,c^*$ and $d^*$ expressed using $A,b,c$ and $d$ since these parameters are linked by the Legendre transformation).

\item Third, the canonical divergences should be considered relative divergences (and not absolute divergences), and defined according to a prescribed arbitrary ``unit'' $\lambda>0$.
Thus we can scale the canonical divergence by $\lambda>0$, i.e., $D_{\lambda,\nabla,\nabla^*}:=\lambda D_{\nabla,\nabla^*}$.
We have $D_{\lambda,\nabla,\nabla^*}(P_1:P_2)=\lambda B_F(\theta_1:\theta_2)=\lambda B_{F^*}(\eta_2:\eta_1)=\lambda Y_F(\theta_1:\eta_2)=\lambda Y_{F^*}(\eta_1:\theta_2)$, and $\lambda B_F(\theta_1:\theta_2)=B_{\lambda F}(\theta_1:\theta_2)$ (and $\nabla \lambda F=\lambda\nabla F$). 
We check the scale invariance of quasi-arithmetic averages: $M_{\lambda \nabla F}=M_{\nabla F}$.
\end{itemize}

Thus we end up with the following invariance and equivariance properties of the quasi-arithmetic averages which have been obtained from an information-geometric viewpoint: 

\begin{proposition}[Invariance and equivariance of quasi-arithmetic averages]\label{prop:equiinvariance}
Let $F(\theta)$ be a  function of Legendre type. Then  $\bar F(\bartheta):=\lambda (F(A\theta+b)+\inner{c}{\theta}+d)$ for $A\in\GL(d)$, $b,c\in\bbR^d$, $d\in\bbR^d$ and $\lambda\in\bbR_{>0}$ is a Legendre-type function, and we have $M_{\nabla \bar F}=A\, M_{\nabla F} + b$.
\end{proposition}

This proposition generalizes Property~\ref{propr:invarqam} of scalar QAMs, and untangles the role of scale $\lambda>0$ from the other invariance roles brought by the Legendre transformation.

\section{Canonical Bregman divergences in dually flat spaces: Legendre affine invariance and divergence unit}\label{sec:leginvar}

By Eguchi contruction~\cite{eguchi1992geometry}, a Bregman divergence $B_F$ induces a unique dually flat space $(M,g,\nabla,\nabla^*)$ with dually flat divergence $D_{\nabla,\nabla^*}$ (a contrast function on the product manifold).
Conversely, we can reconstruct~\cite{IG-2016} a pair of dual potential functions $F(\theta)$ and $G(\eta)$ and their corresponding dual Bregman divergences $B_F$ and $B_G$ from a dually flat space $(M,g,\nabla,\nabla^*)$.
The reconstructed pair of  dual affine coordinate systems $\theta$ and $\eta$ and potential functions $F(\theta)$ and $G(\eta)$  are not unique and related by the Legendre-Fenchel transform (i.e., $G=F^*$).
Indeed, let us define the Bregman generator:
$$
\barF(\theta) = \lambda F(A\theta+b)+ \inner{c}{\theta}+d,
$$
for invertible matrix $A\in\GL(d,\bbR)$, vectors  $b, c\in\bbR^d$ and scalars $d\in\bbR$ and $\lambda\in\bbR_{>0}$.
When function $F(\theta)$ is twice differentiable and strictly convex, so is the function $\barF(\theta)$ since we have
$$
\nabla^2 \barF(\theta)= \lambda A^\top \nabla^2 (\nabla^2 F)(A\theta+b) A \succ 0.
$$ 

The gradient of the generator $\barF$ is
$$
\eta=\nabla\barF(\theta) = \lambda A^\top  {\nabla F}(A\theta+b)+c.
$$

Solving $\nabla\barF(\theta)=\eta$, we get the reciprocal gradient $\theta(\eta)=\nabla\barG(\eta)$: 
$$
\nabla\barG(\eta) = A^{-1}\nabla G\left( \frac{1}{\lambda} A^{-\top} (\eta-c) \right)-b.
$$

The Legendre convex conjugate is obtained as  
\begin{eqnarray*}
\barG(\eta) &=& \inner{\eta}{\nabla\barG(\eta)}-F(\nabla\barG(\eta)),\\
&=& \lambda^\star  G(A^\star\eta+b^\star)+\inner{c^\star}{\eta}+d^\star,
\end{eqnarray*}
with
\begin{eqnarray*}
\lambda^* &=& \lambda,\\
A^\star &=& \frac{1}{\lambda }A^{-1},\\
b^\star &=& -\frac{1}{\lambda}A^{-1}c,\\
c^\star &=& -A^{-1}b,\\
d^\star &=& \inner{b}{A^{-1}c}-d.
\end{eqnarray*}

We checked that  we have:
\begin{eqnarray*}
{\lambda^\star}^\star &=& \lambda,\\
{A^\star}^\star &=& A,\\
{b^\star}^\star &=& b,\\
{c^\star}^\star &=& c,\\
{d^\star}^\star &=& d.
\end{eqnarray*}

That is, the Legendre-Fenchel transform  is an involution.

Notice the interplay of $(A,b)$ with $(c,d)$ when taking the Legendre transform $\calL$.

To summarize, we have:
$$
\calL\left(\lambda F(A\cdot+b)+\inner{c}{\cdot}+d\right)(\eta) 
\xrightarrow{\text{Legendre transform}} 
 \lambda^\star  F^\star(A^\star\eta+b^\star)+\inner{c^\star}{\eta}+d^\star
$$

We check that we have:

\begin{equation}
B_F(\theta_1:\theta_1) = \frac{1}{\lambda}\, B_{\barF}(\bartheta_1:\bar\theta_2),
\end{equation}
where 
$$
\bartheta=A^{-1}(\theta-b).
$$

Geometrically speaking, the torsion-free connection $\nabla$ is flat: That is, there exists a coordinate system $\theta$ such that the Christoffel symbols of $\nabla$ vanish: $\Gamma(\theta)=0$, and hence the $\nabla$-geodesics are line segments in the $\theta$-coordinate system. $\theta$ is called a $\nabla$-affine coordinate system. The coordinate system is not unique as we can choose $\bartheta(p)=A^{-1}(\theta(p)-b)$ as another coordinate system.

Thus we have the dually flat divergence $D_{\nabla,\nabla^*}$ between two points $p_1$ and $p_2$ on $(M,g,\nabla,\nabla^*)$ (with $\theta$-coordinates $\theta_i=\theta(p_i)$ or $\bartheta$-coordinates $\bartheta_i=\bartheta(p_i)$) which can be computed equivalently as follows:
$$
D_{\nabla,\nabla^*}(p_1:p_2) = B_F(\theta_1:\theta_1)=\frac{1}{\lambda}\, B_{\barF}(\bartheta_1:\bar\theta_2),
$$
for any $A\in\GL(d,\bbR)$,   $b, c\in\bbR^d$ and $d, \lambda\in\bbR$.
The scalar $\lambda$ indicates the unit of the dually flat divergence since $\frac{1}{\lambda} B_F=B_{\frac{1}{\lambda} F}$.

\begin{example}
Let us consider the family of categorical distributions on a sample set of size $d$.
That is the family of multinomial distributions with one trial also called sometimes the family of multinoulli distributions.
The order of that exponential family is $D=d-1$.
We have $\theta_i=\log\frac{p_i}{p_d}$ and $F(\theta)=\log(1+\sum_{i=1}^{D}\exp(\theta_i))$ with 
$$
\nabla F(\theta)=\vecthree{ \frac{e^\theta_1}{1+\sum_{j=1}^D e^{\theta_j}} }{\vdots}{  \frac{e^\theta_D}{1+\sum_{j=1}^D e^{\theta_j}}   },
$$ 
and the reciprocal gradient is
$$
\nabla G(\eta)=\vecthree{ \log\frac{\eta_1}{1-\sum_{j=1}^D \eta_j} }{\vdots}{  \log\frac{\eta_D}{1-\sum_{j=1}^D \eta_j} }.
$$
\end{example}

For the special uni-order case of the generators $f(x)$, consider the function
$$
\barf(x) = \lambda f(ax+b)+cx+d,
$$
for $\lambda>0$, $a\not=0$, $c,d\in\bbR$.

Then we have
$$
\barf'(x)=\lambda a f'(ax+b)+c,
$$
and the reciprocal function is found by solving $\barf'(x)=y$:

$$
x(y)=\frac{1}{a}g'\left(\frac{y-c}{\lambda a}\right)-\frac{b}{a}=\barg'(y).
$$

The Legendre convex conjugate is thus 
$$
\barg(y)=x(y)y-\barf(x(y))=
\lambda g\left(\frac{y-c}{\lambda a}\right)-b\frac{y-c}{a}-d.
$$

We check that we have
$$
B_f(x_1:x_2)=\frac{1}{\lambda} B_{\barf}(\barx_1:\barx_2)=B_g(y_2:y_1)=\frac{1}{\lambda} B_{\barg}(\bary_2:\bary_1),
$$
where $\barx=\frac{x-b}{a}$ and
$\bary=\lambda g\left(\frac{y-c}{\lambda a}\right)-b\frac{y-c}{a}-d$.

\begin{example}
Let us consider the Poisson family $\{p_\lambda(x)\ :\ \lambda\in\bbR_{>0}\}$ where $\lambda$ denotes the intensity parameter of a Poisson distribution. 
The natural parameter is $\theta=\log\lambda$, and we get the cumulant function  $F(\theta)=e^\theta$ with $F'(\theta)=e^\theta$, $G'(\eta)=\log\eta$ and convex conjugate $G(\eta)=\eta\log\eta-\eta$.
\end{example}

\section{Quasi-arithmetic statistical mixtures and information geometry}\label{sec:qamJD}

\subsection{Definition of quasi-arithmetic statistical mixtures}

Consider a quasi-arithmetic mean $m_f$.
We consider $n$ probability distributions $P_1,\ldots, P_n$ all dominated by a measure $\mu$, and denote by 
$p_1=\frac{\mathrm{d}P_1}{\dmu},\ldots, p_n=\frac{\mathrm{d}P_n}{\dmu}$ their Radon-Nikodym derivatives.
Let us define statistical $m_f$-mixtures of $p_1,\ldots, p_n$:

\begin{definition}\label{def:qamix}
The $m_f$-mixture of $n$ densities $p_1,\ldots, p_n$ weighted by $w\in\Delta_{n}^\circ$ is defined by
$$
(p_1,\ldots,p_n;w)^{m_f}(x):= \frac{m_f(p_1(x),\ldots,p_n(x);w)}{\int m_f(p_1(x),\ldots,p_n(x);w)\dmu(x)}.
$$
\end{definition}

The quasi-arithmetic mixture (QAMIX for short) $(p_1,\ldots,p_n;w)^{m_f}$ generalizes the ordinary statistical mixture $\sum_{i=1}^d w_i p_i(x)$ when $f(t)=t$ and $m_f=A$ is the arithmetic mean. A statistical $m_f$-mixture can be interpreted as the $m_f$-integration of its weighted component densities, the densities $p_i$'s.
The power mixtures $(p_1,\ldots,p_n;w)^{m_p}(x)$ (including the ordinary and geometric mixtures) are called $\alpha$-mixtures 
in~\cite{Integration-Amari-2007} with $\alpha(p)=1-2p$ (or $p=\frac{1-\alpha}{2}$).
A nice characterization of the $\alpha$-mixtures is that these mixtures are the density centroids of the weighted mixture components with respect to  the $\alpha$-divergences~\cite{Integration-Amari-2007} (proven by calculus of variation):
$$
(p_1,\ldots,p_n;w)^{m_{\alpha}}=\arg\min_p w_i D_\alpha(p_i,p),
$$
where $D_\alpha$ denotes the $\alpha-$divergences~\cite{IG-2016,nielsen2014clustering}.
$m_f$-mixtures have also been used to define a generalization of the Jensen-Shannon divergence~\cite{JS-2019} between densities $p$ and $q$ as follows:
\begin{equation}
D_\JS^{m_f}(p,q):=\frac{1}{2}\left(D_\KL(p:(pq)^{m_f})+D_\KL(q:(pq)^{m_f})\right)\geq 0,
\end{equation}
where $D_\KL(p:q)=\int p(x)\log \frac{p(x)}{q(x)}\dmu(x)$ is the Kullback-Leibler divergence, and $(pq)^{m_f}:=(p,q;\frac{1}{2},\frac{1}{2})^{m_f}$. The ordinary JSD is recovered when $f(t)=t$ and $m_f=A$:
$$
D_\JS(p,q)=\frac{1}{2}\left(D_\KL\left(p:\frac{p+q}{2}\right)+D_\KL\left(q:\frac{p+q}{2}\right)\right).
$$
Quasi-arithmetic mixtures of two components have also been used to upper bound the probability of error in Bayesian hypothesis testing~\cite{nielsen2014generalized}.

Let us give some examples of parametric families of probability distributions that are closed under quasi-arithmetic mixturing:

\begin{itemize}
\item
Consider a natural exponential family~\cite{BN-2014} $\calE=\{p_\theta=\exp(\theta\cdot x-F(\theta))\st \theta\in\Theta \}$.
Function $F(\theta)=\log\int \exp(\theta\cdot x)\dmu$ is called the cumulant function (the log-normalizer function called  log-partition in statistical physics), and is of Legendre type when the exponential family is steep~\cite{BN-2014}.
Regular exponential families with are full exponential families with open natural parameter spaces $\Theta$ are steep.
We have Shannon entropy of density $p_\theta\in\calE$ expressed using the negative convex conjugate~\cite{HEF-2010} which is concave: 
$H(p_\theta)=-F^*(\eta)$ with $\eta=\nabla F(\theta)$.
Since exponential families are closed under geometric mixtures, i.e. $(p_{\theta_1}p_{\theta_2})^{G}=p_{\frac{\theta_1+\theta_2}{2}}$, we have Shannon entropy which can be expressed using the convex conjugate: 
$$
H((p_{\theta_1}p_{\theta_2})^{G})=-F^*\left(\nabla F \left(\frac{\theta_1+\theta_2}{2}\right)\right).
$$
We can rewrite $\nabla F \left(\frac{\theta_1+\theta_2}{2}\right)$ as the quasi-arithmetic average $M_{\nabla F^*}(\eta_1,\eta_2)$.
More generally, the geometric mixtures of $n$ densities of an exponential family belongs to that exponential family:
$$
(p_{\theta_1},\ldots,p_{\theta_n};w)^{G}\propto \prod_{i=1}^n p_{\theta_i}^{w_i}  = p_{\sum_{i=1}^n w_i\theta_i}.
$$
That is the normalization constant of $\prod_{i=1}^n p_{\theta_i}^{w_i}$ is $\exp(F(\sum_i w_i F(\theta_i))-\sum w_i F(\theta_i)=\exp(-J_F(\theta_1,\ldots,\theta_n;w))$, where $J_F$ is called the Jensen diversity index~\cite{burbea1982convexity,BR-2011}.

We may also build an exponential family by considering $n+1$ probability distributions $P_0, \ldots, P_n$ mutually absolutely continuous and all dominated by a reference measure $\mu$. Let $p_0,\ldots, p_n$ denote their Radon-Nikodym densities  such that  $\log\frac{p_1}{p_0},\ldots, \log\frac{p_{n}}{p_0}$ are linear independent.
Then consider the geometric mixture family:
$$
\calG=\left\{ (p_{0},\ldots,p_{n};w)^{G}= \st w\in\Delta_{n}^\circ \right\}.
$$
We have $(p_{0},p_1,\ldots,p_{n};w)^{G}\propto 
\exp\left(\sum_{i=1}^n w_i\log \frac{p_i}{p_0}\right) p_0(x)$.
We let the natural parameter be $\theta=(w_1,\ldots,w_n)\in \Delta_{n}^\circ$.
When $n=1$, Gr\"unwald~\cite{grunwald2007minimum} called $\calG$ a likelihood ratio exponential family (LREF).
Uni-order LREFs ($n=1$) have been studied in~\cite{brekelmans2020likelihood,nielsen2022revisiting}: It has the advantage of considering a non-parametric statistical model~\cite{zhang2013nonparametric} as a 1D exponential family model which yields a convenient framework
 for studying the statistical model $\calG$ under the lens of well-studied exponential families~\cite{BN-2014}.

\begin{remark}
We may consider another equivalent definition of the ordinary JSD~\cite{Sibson-1969,Lin-1991}  given  by
\begin{equation}
D_\JS(p,q) = H\left(\frac{p+q}{2}\right)-\frac{H(p)+H(q)}{2}\geq 0,\label{eq:ojsd}
\end{equation}
where  $H(p)=-\int p(x)\log p(x)\dmu(x)$ is the strictly concave Shannon entropy.
Thus we may consider the following generalization of the JSD:
\begin{equation}
H_\JS^{M_f,M_g}(p,q) = H\left((pq)^{m_f}\right)-m_g(H(p),H(q)),\label{eq:hjsd}
\end{equation}
where $m_f$ and $m_g$ are two quasi-arithmetic means. 
The first QAM is used to build a quasi-arithmetic mixture while the second QAM is used to average scalars.
When $f(t)=g(t)=t$, we recover the ordinary JSD with $m_f=m_g=A$.
Let us introduce the $(M,N)$-Jensen divergences~\cite{MNJensenBregman-2017} according to two generic symmetric bivariate means $M$ and $N$:
\begin{equation} 
J_F^{M,N}(\theta_1,\theta_2) = M(F(\theta_1),F(\theta_2))-F(N(\theta_1,\theta_2)).
\end{equation}
We recover the ordinary Jensen divergence~\cite{burbea1982convexity,BR-2011} when $M=N=A$:
$$
J_F(\theta_1,\theta_2)=J_F^{A,A}(\theta_1,\theta_2)=\frac{F(\theta_1)+F(\theta_2)}{2}-F\left(\frac{\theta_1+\theta_2)}{2}\right).
$$
Jensen divergences are non-negative and equal to zero only when $\theta_1=\theta_2$ when $F$ is a strictly convex function.
Similarly, by definition, $J_F^{M,N}(\theta_1,\theta_2)\geq 0$ with equality only when $\theta_1=\theta_2$ when $F$ is said $(M,N)$-convex~\cite{meanconvexity-2003} and~\cite{ConvexBook-2006} (Appendix~A).
A $(G,A)$-convex function is said log-convex and a $(G,G)$-convex function is said multiplicative convex.
We can test whether a function $g$ is strictly $(m_{f_1},m_{f_2})$-convex by checking whether  the function $f_2\circ g\circ f_1$ is strictly convex or not  (see the correspondence Lemma A.2.2 in~\cite{ConvexBook-2006}).
For densities $p_{\theta_1}$ and $p_{\theta_2}$ belonging to a same exponential family $\calE$, we have
\begin{eqnarray*}
H_\JS^{G,A}(p_{\theta_1},p_{\theta_2}) &=& H((p_{\theta_1}p_{\theta_2})^G)-A(H(p_{\theta_1}),H(p_{\theta_2})),\\
&=& -F^*\left(M_{\nabla F^*}(\eta_1,\eta_2)\right)
+\frac{F^*(\eta_1)+F^*(\eta_2)}{2}
=J^{\nabla F^*,A}_{F^*}(\eta_1,\eta_2),
\end{eqnarray*}
where $J^{\nabla F^*,A}$ is the $(\nabla F^*,A)$-Jensen divergence defined according to two means.
Thus we have $H_\JS^{G,A}(p_{\theta_1},p_{\theta_2})\geq 0$ iff $F^*$ is $(\nabla F^*,A)$-convex.

To see how $H_\JS^{G,A}$ differs from $D_\JS^{G}$ defined in~\cite{JS-2019}, 
let us introduce the cross-entropy between $p$ and $q$: $H^\times(p:q)=-\int p(x)\log q(x)\dmu(x)$.
Then $D_\KL(p:q)=H^\times(p:q)-H(p)$ with $H(p)=H^\times(p:p)$, and we have
\begin{eqnarray}
D_\JS^{G}(p:q) &=& \frac{1}{2}\left(D_\KL(p:(pq)^G)+ D_\KL(q:(pq)^G)\right)\geq 0,\\
&=& \frac{1}{2}\left(H^\times(p:(pq)^G)-H(p)+ H^\times(q:(pq)^G)-H(q)\right)\geq 0,\\
&=& H^\times((pq)^A:(pq)^G)-\frac{H(p)+H(q)}{2}\geq 0.
\end{eqnarray}
However, we have $H_\JS^{G,A}(p:q)=H^\times((pq)^G:(pq)^G)-\frac{H(p)+H(q)}{2}$.
Therefore, the dissimilarity  $H_\JS^{G,A}(p:q)$ can be potentially negative when $H^\times((pq)^G:(pq)^G)\leq H^\times((pq)^A:(pq)^G)$.
\end{remark}

\item
Let $p_0, p_1,\ldots, p_n$ denotes $n+1$ linearly independent densities, and consider their (arithmetic/standard) mixture family~\cite{IG-2016}:
$\calM=\{m_{\theta}(x)=\sum_{i=0}^n w_i p_i(x) \st w\in\Delta_{n}^\circ\}$ 
with $\theta=(w_1,\ldots, w_{n})\in\Delta_{n-1}^\circ$ (and $w_0=1-\sum_{i=1}^{1}\theta_i$). 
The Shannon negentropy $F(\theta)=-H(m_\theta)$ is a Legendre type function~\cite{MCIG-2019}.
Since the mixture of two densities of a mixture family $\calM$ belongs to $\calM$ (i.e., $\frac{m_{\theta_1}+m_{\theta_2}}{2}=m_{\frac{\theta_1+\theta_2}{2}}$), we have $H\left(\frac{m_{\theta_1}+m_{\theta_2}}{2}\right)=H\left(m_{\frac{\theta_1+\theta_2}{2}}\right)=-F\left({\frac{\theta_1+\theta_2}{2}}\right)$.
It follows that $D_\JS(m_{\theta_1}:m_{\theta_2})=J_F(\theta_1:\theta_2)\geq 0$.

\item Consider the family of scale Cauchy distributions $
\calC=\left\{ p_s(x)= \frac{1}{\pi s} \frac{1}{1+\left(\frac{x}{s}\right)^2} \st s\in\bbR_{>0}  \right\}$.
The harmonic mixture $(p_{s_1}p_{s_2})^H$ of two Cauchy distributions is a Cauchy distribution $p_{s_{12}}$~\cite{nielsen2014generalized} 
with parameter $s_{12}=\sqrt{\frac{s_1s_2^2+s_2s_1^2}{s_1+s_2}}$.
More generally, the harmonic mixtures of $n$ scale Cauchy distributions is a Cauchy distribution.

\item The power mixture of central multivariate $t$-distributions is a central multivariate $t$-distribution~\cite{nielsen2014generalized} .
\end{itemize}
\begin{figure}[hbtp]
\centering

\includegraphics[width=0.7\textwidth]{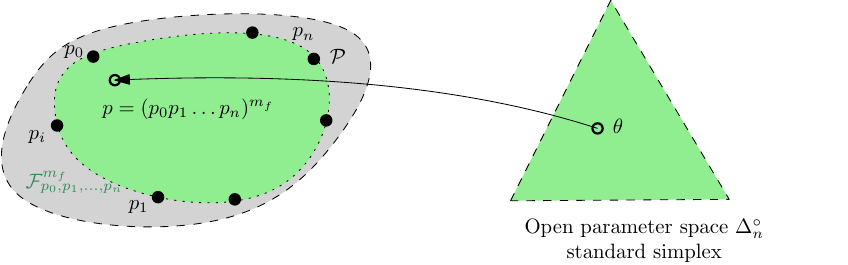}

\caption{Statistical models in the quasi-mixture family are parametrized by a vector in the open standard simplex.\label{fig:qamixs}}
\end{figure}

In general, we may consider quasi-arithmetic paths between densities on the space $\calP$ of probability density functions with a common support all dominated by a reference measure.
On $\calP$, we can build a parametric statistical model called a $m$-mixture family of order $n$ as follows:
$$
\calF_{p_0,p_1,\ldots,p_n}^{m_f}:=\left\{ 
(p_0,p_1,\ldots,p_n;(\theta,1))^{m_f} \st \theta\in\Delta^\circ_n \right\}.
$$
In particular, power $q$-paths have been investigated in~\cite{masrani2021q} with applications in annealing importance sampling and other Monte Carlo methods.
The information geometry of such a density space with quasi-arithmetic paths has been investigated in~\cite{eguchi2015path} by
  considering quasi-arithmetic means with respect to a monotone increasing and concave function. See also~\cite{vigelis2017existence,eguchi2018information}.
	
\subsection{The $\nabla$-Jensen-Shannon divergences}\label{sec:nablaJS}
	
We conclude by giving a geometric definition of a generalization of the Jensen-Shannon divergence on $\calP$ according to an arbitrary affine connection~\cite{IG-2016,zhang2013nonparametric} $\nabla$:

\begin{definition}[Affine connection-based $\nabla$-Jensen-Shannon divergence]\label{def:geojsd}
Let $\nabla$ be an affine connection on the space of densities $\calP$, and $\gamma_\nabla(p,q;t)$ the geodesic linking density $p=\gamma_\nabla(p,q;0)$ 
to density $q=\gamma_\nabla(p,q;1)$. Then the $\nabla$-Jensen-Shannon divergence is defined by:
\begin{equation}\label{eq:geojsd}
D^\JS_\nabla(p,q) := \frac{1}{2}\left( D_\KL\left(p:\gamma_\nabla\left(p,q;\frac{1}{2}\right)\right) + D_\KL\left(q:\gamma_\nabla\left(p,q;\frac{1}{2}\right)\right) \right).
\end{equation}
\end{definition}	

When $\nabla=\nabla^m$ is chosen as the mixture connection~\cite{IG-2016}, we end up with the ordinary Jensen-Shannon divergence since 
$\gamma_{\nabla^m}(p,q;\frac{1}{2})=\frac{p+q}{2}$. When $\nabla=\nabla^e$, the exponential connection, we get the geometric Jensen-Shannon divergence~\cite{JS-2019} since $\gamma_{\nabla^e}(p,q;\frac{1}{2})=(pq)^G$ is a statistical geometric mixture.
We may choose the $\alpha$-connections of information geometry to define $\nabla$-Jensen-Shannon divergences (see Figure~\ref{fig:geomidpoints}).

\begin{figure}[hbtp]
\centering

\includegraphics[width=0.3\textwidth]{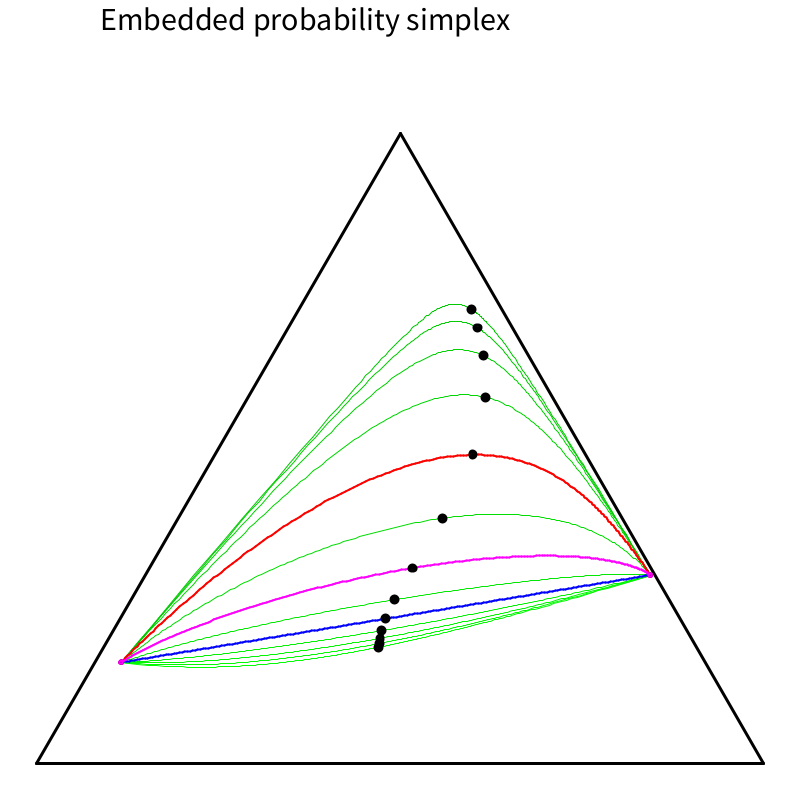}
\includegraphics[width=0.3\textwidth]{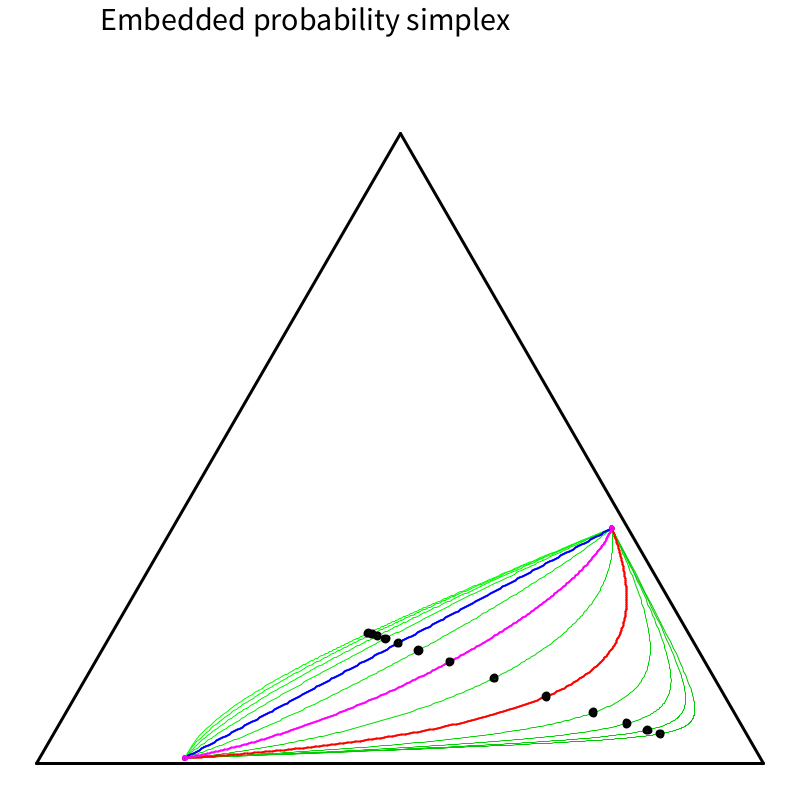}
\includegraphics[width=0.3\textwidth]{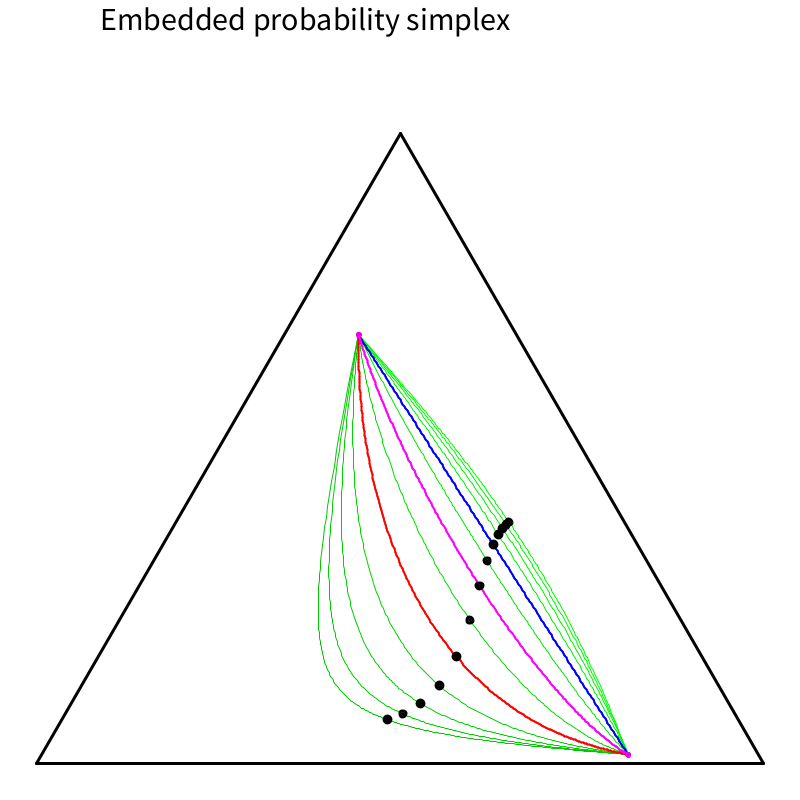}\\
\includegraphics[width=0.3\textwidth]{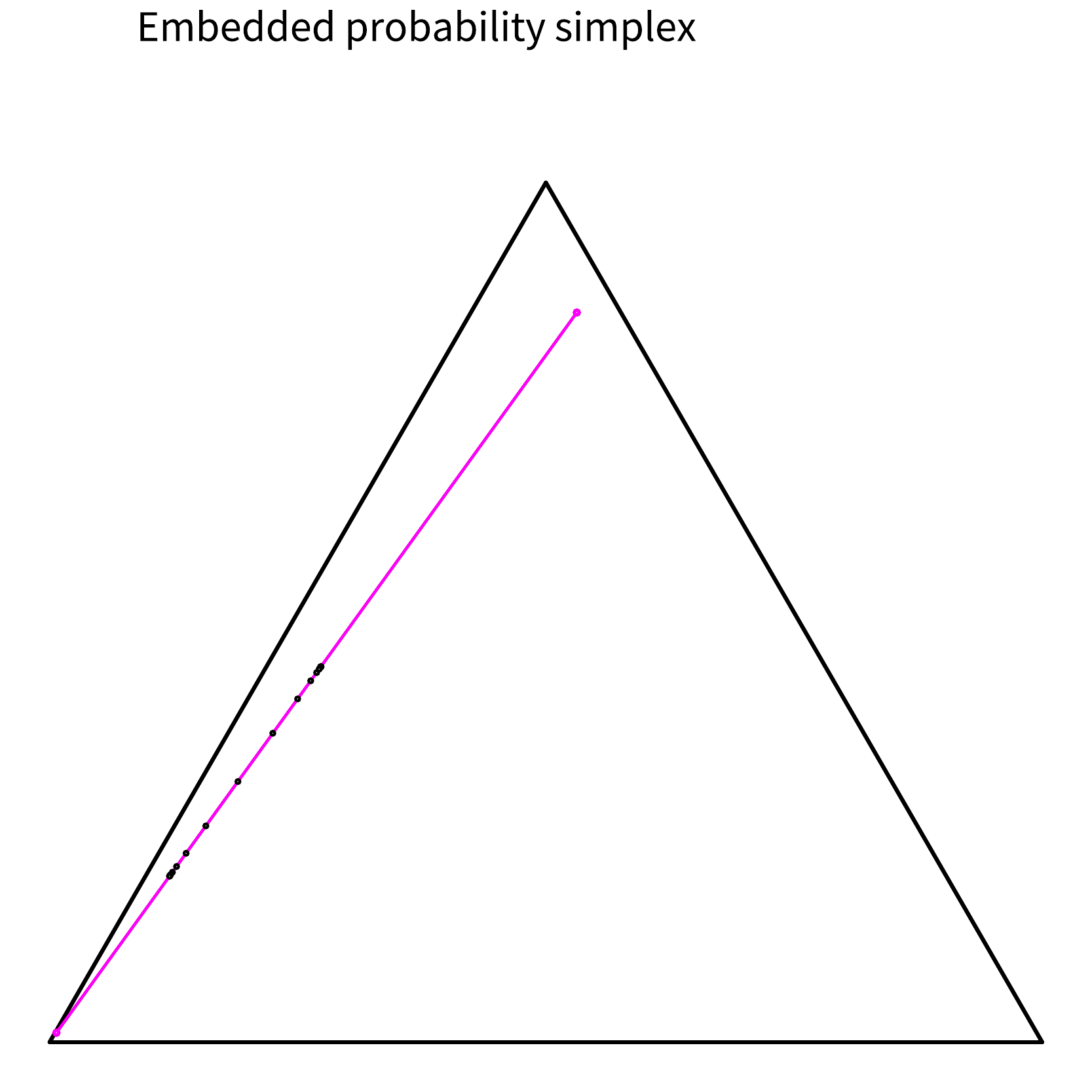}
\includegraphics[width=0.3\textwidth]{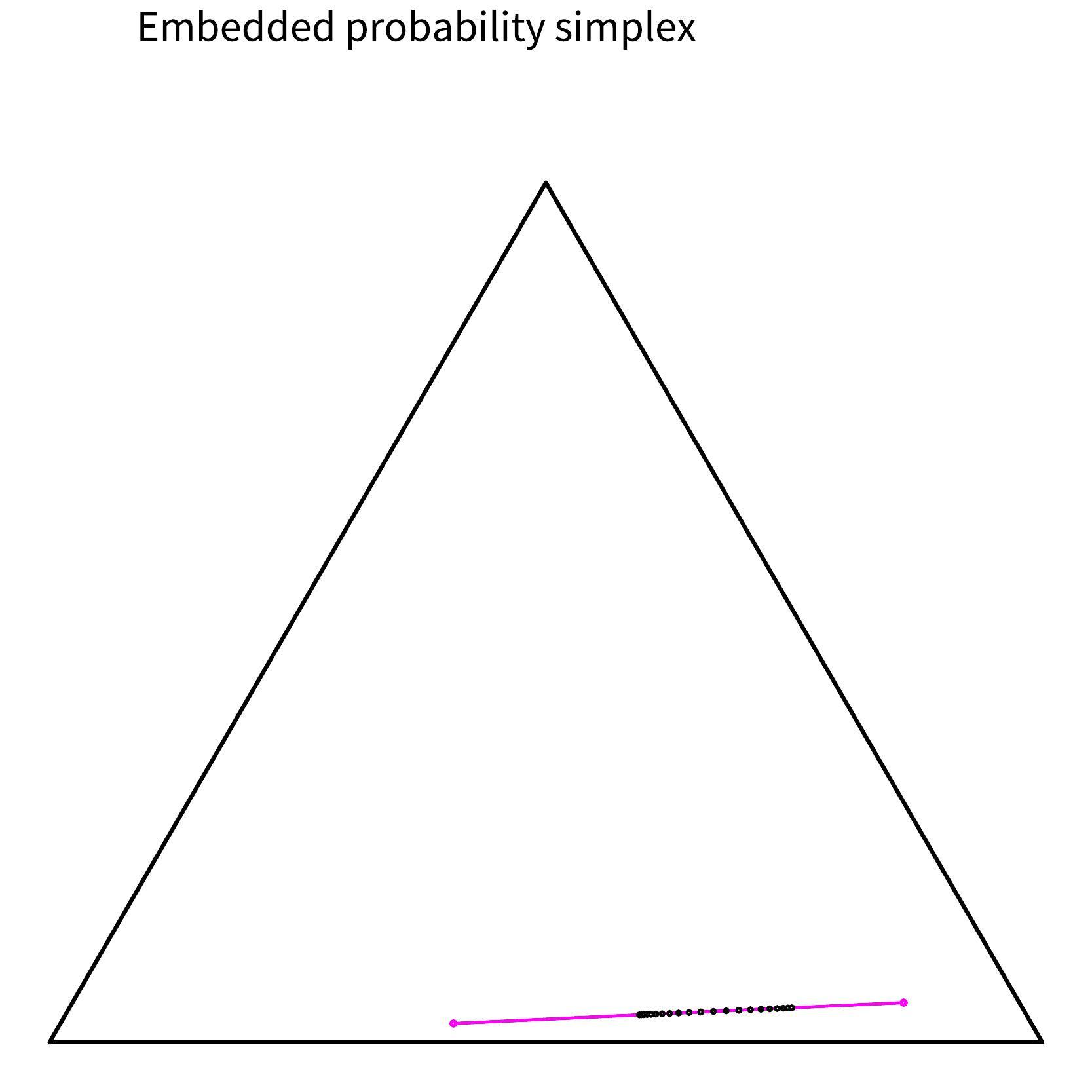}
\includegraphics[width=0.3\textwidth]{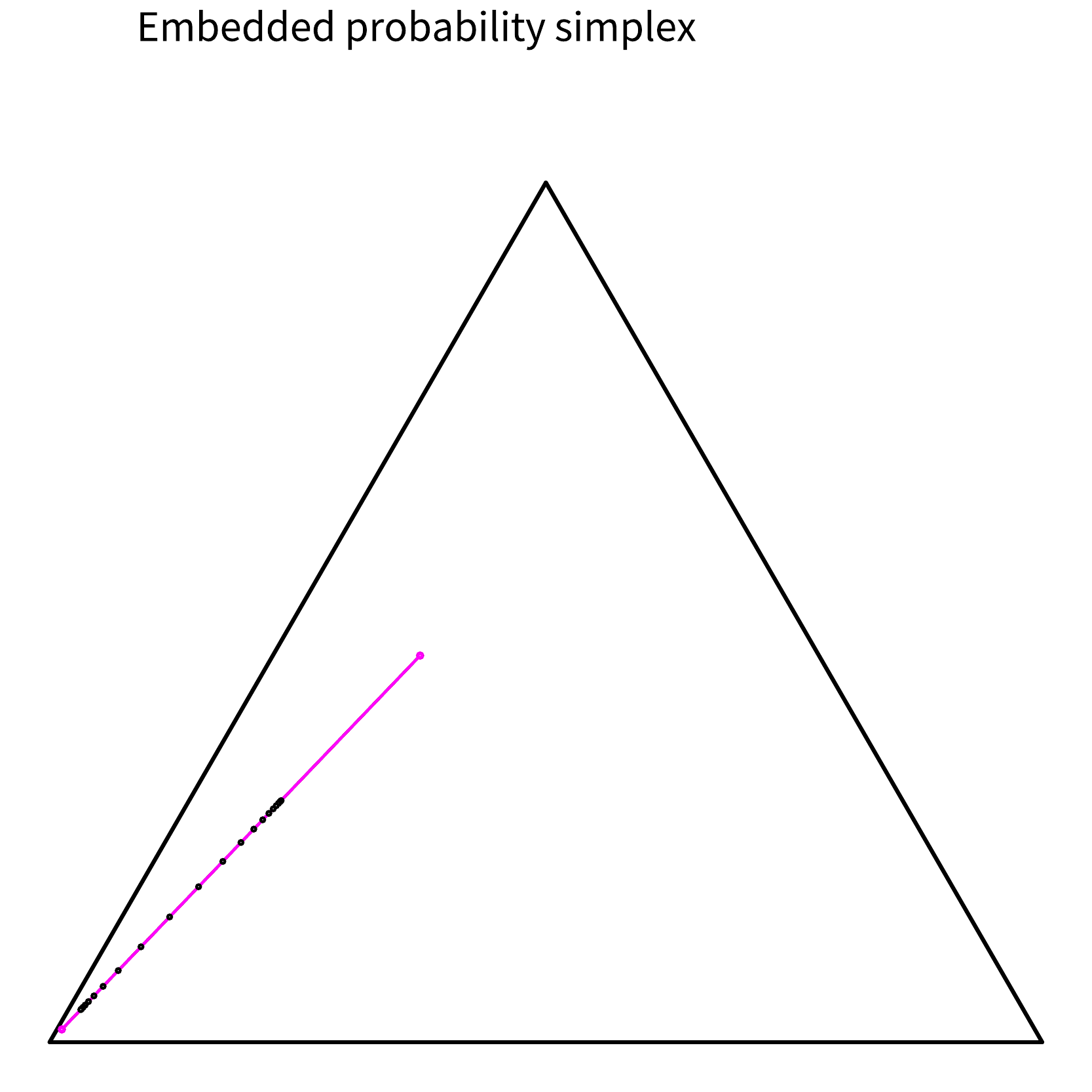} 

\caption{Top: Some $\alpha$-geodesics rendered in the 2D probability simplex (equilateral triangle sitting in 3D) with their midpoints displayed.
Bottom: When the points $p$ and $q$ are collinear with a vertex of the probability simplex, the $\alpha$-geodesics coincide with the line $(pq)$.
\label{fig:geomidpoints}}
\end{figure}

When the space of densities $\calP$ is a exponential family or a mixture family with carrying a dually flat structure $(\calP,g_{\mathrm{Fisher}},\nabla^m,\nabla^{e})$ where $g_{\mathrm{Fisher}}$ denotes the Riemannian Fisher information metric~\cite{IG-2016}, we have the Kullback-Leibler divergence $D_\KL$ which can be expressed using the canonical divergence $D_{\nabla^m,\nabla^{e}}$, and the Jensen-Shannon divergence can be written geometrically as
$$
D_\JS(p:q)=D^\JS_\nabla(P,Q):=\frac{1}{2}\left(D_{\nabla^m,\nabla^{e}}\left(p:\gamma_{\nabla^m}(p,q;\frac{1}{2})\right) + 
D_{\nabla^m,\nabla^{e}}\left(q:\gamma_{\nabla^m}(p,q;\frac{1}{2})\right) \right),
$$
where $P$ and $Q$ denote the points on $\calP$ representing the densities $p$ and $q$.

Furthermore, we may consider the $\alpha$-connections~\cite{IG-2016} $\nabla^\alpha$ of parametric or non-parametric statistical models, and skew the geometric Jensen-Shannon divergence to define the $\beta$-skewed $\nabla^\alpha$-JSD:
$$
D^\JS_{\nabla^\alpha,\beta}(p,q)=
\beta D_\KL(p:\gamma_{\nabla^\alpha}(p,q;\beta)) + (1-\beta) D_\KL(q:\gamma_{\nabla^\alpha}(p,q;\beta)).
$$

\section{Concluding remarks}

In this paper, we presented two generalizations of the scalar quasi-arithmetic means~\cite{Inequalities-1952} through the lens of information geometry, and discussed some of their applications:

\begin{itemize}
	\item The first generalization  of  scalar quasi-arithmetic means consisted in defining pairs of {\em quasi-arithmetic averages} induced by the gradient maps of pairs of Legendre-type functions.
	These dual quasi-arithmetic averages are used in information geometry to express points on dual geodesics and sided barycenters in the dually affine 
	$\theta$- and $\eta$-coordinate systems. 
	Furthermore, we proved that $M_{\nabla F}=M_{\nabla \bar F}=M_{[\nabla F]}$ where $[\nabla \bar F]$ denotes the equivalence class of Legendre type functions such that $\bar F(\bar\theta)=\lambda F(\theta+b)+\inner{c}{\theta}+d \sim F(\theta)$.
	This property generalizes the well-known fact that  quasi-arithmetic means $M_f=M_g$ iff $g=\lambda f+c$ and distinguishes the  scaling invariance by $\lambda>0$ with the Legendre invariance by $c$.
	
	\item The second generalization of quasi-arithmetic means defined {\em statistical quasi-arithmetic mixtures} by normalizing   quasi-arithmetic means of their densities:
	In particular, we showed how  exponential families are closed under geometric mixtures, and described a generic way to build exponential families of order $n$ from geometric mixtures of $n+1$ linear independent log ratio densities $\log \frac{p_i}{p_0}$.
	The  statistical geometric mixture family  construction holds similarly for other quasi-arithmetic mixture families. 
	Last, we gave a generic geometric definition of the Jensen-Shannon divergence based on affine connections which generalizes both the ordinary Jensen-Shannon divergence~\cite{Lin-1991} and the geometric Jensen-Shannon divergence~\cite{JS-2019}.
	This demonstrates the rich interplay of divergences with information geometry.
\end{itemize}

 \bibliographystyle{plain}
 \bibliography{MNJSDBibV4}
 
\end{document}